\documentclass[submission,copyright,creativecommons]{eptcs}
\providecommand{\event}{ACT 2021} %
\pdfoutput=1
\usepackage{underscore}           %

\usepackage{graphicx}
\usepackage{subfigure}
\usepackage[table, dvipsnames]{xcolor}
\usepackage{amsmath,amssymb}
\usepackage{mathtools}
\usepackage{wrapfig}
\usepackage{longtable}
\usepackage{booktabs}
\usepackage{pythonhighlight}
\usepackage{array}
\usepackage{pgffor}
\usepackage{fontawesome}
\usepackage{stmaryrd}
\usepackage{eucal}
\usepackage[utf8]{inputenc}

\usepackage{amsthm}
\usepackage{mathtools}

\theoremstyle{plain}
\newtheorem{lemma}{Lemma}[section]

\newtheorem{proposition}[lemma]{Proposition}
\newtheorem{theorem}[lemma]{Theorem}

\theoremstyle{definition}
\newtheorem{definition}[lemma]{Definition}

\usepackage{blanko}

\makeatletter
\def\bstr{b}
\def\bfstr{bf}
\def\cstr{c}
\def\fstr{f}
\def\strLst{A,B,C,D,d,E,F,G,H,I,J,K,L,M,N,O,P,Q,R,S,T,U,V,W,X,Y,Z}

\newcommand{\MkB}[1]{\expandafter\def\csname\bstr#1\endcsname{\mathbb{#1}}}
  \@for\i:=\strLst\do{%
    \expandafter\MkB \i     }

\newcommand{\MkBF}[1]{\expandafter\def\csname\bfstr#1\endcsname{\mathbf{#1}}}
  \@for\i:=\strLst\do{%
    \expandafter\MkBF \i     }

\newcommand{\MkCal}[1]{\expandafter\def\csname\cstr#1\endcsname{\mathcal{#1}}}
  \@for\i:=\strLst\do{%
    \expandafter\MkCal \i     }

\newcommand{\MkFrak}[1]{\expandafter\def\csname\fstr#1\endcsname{\mathfrak{#1}}}
  \@for\i:=\strLst\do{%
    \expandafter\MkFrak \i     }
    
\makeatother

\newcommand{\Lin}[1]{\mathop{\mathsf{Lin}}(#1)}

\newcommand{\RMatchGT}[3]{\mathcal{M}^{{\text{\tiny $#1$}}}_{#2}(#3)}

\newcommand{\obj}[1]{\mathsf{obj}(#1)}

\newcommand{\mIO}{\mathop{\varnothing}}

\newcommand{\tMatchGT}[3]{\mathsf{MT}^{#3}_{#1}(#2)}

\newcommand{\TcompGT}[4]{#1\;{}^{#2}\!\!\!{\angle}_{#4}\; #3}

\newcommand{\tapGT}[3]{#1 \diamond_{#3} #2}

\makeatletter
\newcommand{\ostar}{\mathbin{\mathpalette\make@circled\star}}
\newcommand{\make@circled}[2]{%
  \ooalign{$\m@th#1\smallbigcirc{#1}$\cr\hidewidth$\m@th#1#2$\hidewidth\cr}%
}
\newcommand{\smallbigcirc}[1]{%
  \vcenter{\hbox{\scalebox{0.77778}{$\m@th#1\bigcirc$}}}%
}
\makeatother

\newcommand{\actto}{\rightarrow\Mapsfromchar}

\DeclareMathAlphabet{\mathbbe}{U}{bbold}{m}{n}
\newcommand{\simplexcategory}{\mathbbe{\Delta}}

\newcommand{\op}{^{\text{{\rm{op}}}}}

\DeclareRobustCommand\upperstar{%
  \mathchoice%
    {\kern0pt\raise0.55ex\hbox{$\displaystyle *$}\kern0.8pt}%
    {\kern0pt\raise0.58ex\hbox{$\textstyle *$}\kern0.8pt}%
    {\kern0pt\raise0.45ex\hbox{$\scriptstyle *$}\kern0.4pt}%
    {\kern0pt\raise0.4ex\hbox{$\scriptscriptstyle *$}\kern0.2pt}%
}%
\DeclareRobustCommand\lowerstar{%
  \mathchoice%
    {\kern0pt\raise-0.65ex\hbox{$\displaystyle *$}\kern0.8pt}%
    {\kern0pt\raise-0.68ex\hbox{$\textstyle *$}\kern0.8pt}%
    {\kern0pt\raise-0.55ex\hbox{$\scriptstyle *$}\kern0.4pt}%
    {\kern0pt\raise-0.5ex\hbox{$\scriptscriptstyle *$}\kern0.2pt}%
}%

\newcommand{\lowershriek}{_!}

\def\overarrow#1{{\vec{#1}}}
\def\nondeg{\overarrow}

\newcommand{\isopil}{\stackrel{\raisebox{0.1ex}[0ex][0ex]{\(\sim\)}}%
			{\raisebox{-0.15ex}[0.28ex]{\(\rightarrow\)}}}

\providecommand{\kat}[1]{\text{\textbf{\textsl{#1}}}}

\newcommand{\Grpd}{\kat{Grpd}}

\newcommand{\N}{\mathbb{N}}

\newcommand{\SSS}{\mathsf{S}}

\DeclareRobustCommand\upperstar{%
  \mathchoice%
    {\kern0pt\raise0.55ex\hbox{$\displaystyle *$}\kern0.8pt}%
    {\kern0pt\raise0.58ex\hbox{$\textstyle *$}\kern0.8pt}%
    {\kern0pt\raise0.45ex\hbox{$\scriptstyle *$}\kern0.4pt}%
    {\kern0pt\raise0.4ex\hbox{$\scriptscriptstyle *$}\kern0.2pt}%
}%
\DeclareRobustCommand\lowerstar{%
  \mathchoice%
    {\kern0pt\raise-0.65ex\hbox{$\displaystyle *$}\kern0.8pt}%
    {\kern0pt\raise-0.68ex\hbox{$\textstyle *$}\kern0.8pt}%
    {\kern0pt\raise-0.55ex\hbox{$\scriptstyle *$}\kern0.4pt}%
    {\kern0pt\raise-0.5ex\hbox{$\scriptscriptstyle *$}\kern0.2pt}%
}%

\def\overarrow#1{{\vec{#1}}}
\def\nondeg{\overarrow}

\providecommand{\kat}[1]{\text{\textbf{\textsl{#1}}}}

\newcommand{\X}{\mathbf{X}}
\newcommand{\Y}{\mathbf{Y}}
\newcommand{\Z}{\mathbf{Z}}

\newcommand{\ti}[1]{%
 \ensuremath{\vcenter{\hbox{\includegraphics{diagrams/#1.pdf}}}}%
}

\colorlet{h1color}{blue!70!black} %
\colorlet{h2color}{orange!90!black} %
\colorlet{h3color}{blue!40!white} %
\colorlet{h4color}{green!40!black} %

\definecolor{qBlue}{RGB}{92,92,214}
\definecolor{qGreen}{RGB}{92,214,92}  %

\title{Tracelet Hopf Algebras and Decomposition spaces\\(Extended 
Abstract)}
\author{Nicolas Behr
\institute{Universit\'{e} Paris Cit\'{e}, CNRS, IRIF}
\email{nicolas.behr@irif.fr}
\and
Joachim Kock\thanks{Supported by grants MTM2016-80439-P (AEI/FEDER, UE) of Spain and
2017-SGR-1725 of Catalonia.}
\institute{Universitat Aut\`onoma de Barcelona \\ \& Centre de Recerca 
Matem\`atica}
\email{kock@mat.uab.cat}
}
\def\titlerunning{Tracelet Hopf Algebras and Decomposition Spaces}
\def\authorrunning{N. Behr and J. Kock}
\begin{document}
\maketitle

\begin{abstract}
Tracelets are the intrinsic carriers of causal information in
categorical rewriting systems. In this work, we assemble tracelets
into a symmetric monoidal decomposition space, inducing a
cocommutative Hopf algebra of tracelets. This Hopf algebra captures
important combinatorial and algebraic aspects of rewriting theory,
and is motivated by applications of its representation theory to
stochastic rewriting systems such as chemical reaction networks.
\end{abstract}

\section{Introduction}

Double-Pushout (DPO)~\cite{ehrig:2006aa} and more generally compositional 
categorical rewriting systems~\cite{nbSqPO2019, Behr2021compositionality} provide a 
versatile and mathematically sound framework for modeling complex transition 
systems, with a paradigmatic example being the modeling of reaction systems in biochemistry~\cite{Boutillier:2018aa} and in organic chemistry~\cite{BKAM2021}. %
The specification of an individual rewriting operation (\emph{direct derivation}) in essence amounts to providing a rewrite rule, i.e., a span of monomorphisms, that acts as a sort of template for the operation, together with a \emph{match}, which permits to specify the location within a host object where the local replacement operation is to be performed. In practical applications, it is often the case that the rewriting rules themselves involve only comparatively small graph-like objects. In contrast, the host objects to which the rewrites are applied could easily be several orders of magnitude larger, so that an enormous number of matches may be possible for a given rule and a given host object. 

A natural and powerful approach to overcome this fundamental problem consists in focusing on the combinatorial, statistical and structural properties of \emph{interactions} of rewriting rules within derivation traces, and to aim for a classification of traces in terms of ``interaction patterns''. %
Unlike in compositional diagrammatic calculi such as in particular the theory of string diagrams, the key obstacle for such a type of analysis in rewriting theories resides in the fact that two given rules may in general interact in a multitude of ways, i.e., there does not exist a notion of deterministic rule composition. %
Instead, as first demonstrated in~\cite{bdg2016} and further 
developed in~\cite{nbSqPO2019,bp2019-ext,BKAM2021, 
Behr2021compositionality}, it is necessary to define a notion of 
non-deterministic rule composition via a form of recursive 
application of the concurrency theorem, which then indeed permits
to reason statically about classes of rule compositions. %
Taking inspiration from the notion of \emph{pathways} in chemical reaction systems, this approach was then further refined in~\cite{behr2019tracelets} to the notion of \emph{tracelets}, which in essence act as the carriers of causal information in derivation sequences. %

The main objective of the present paper is to establish a principled mathematical approach to formalize the \emph{combinatorics} of tracelets. Generalizing results of~\cite{bdgh2016} on rewriting systems over directed multi-graphs to the categorical rewriting theory setting, we will demonstrate that it is indeed the notion of \emph{combinatorial Hopf algebras} that naturally captures the rich structure of tracelets. Apart from a rewriting-theoretic construction of the Hopf algebras (Section~\ref{sec:Hopf}), we report on our original discovery that this at first sight seemingly ad hoc construction is in fact interpretable in terms of the theory of \emph{decomposition spaces} (Sections~\ref{sec:DSRR} and~\ref{sec:DST}). Our motivation for this approach has been the analogy between the inductive definition of tracelets and the decomposition-space axioms in homotopy combinatorics. The benefit
of this `detour' is to situate the tracelet Hopf algebra in a
general framework covering most combinatorial Hopf algebras, thereby
exhibiting the constructions and proofs 
as instances of general ideas.

Decomposition spaces were introduced in combinatorics by G\'alvez,
Kock and Tonks~\cite{Galvez-Kock-Tonks151207573, Galvez-Kock-Tonks151207577} as a far-reaching homotopical generalization of posets for the purpose of incidence
algebras and M\"obius inversion, and independently by Dyckerhoff and
Kapranov~\cite{Dyckerhoff-Kapranov12123563} in homological algebra and representation theory, motivated mainly by the Waldhausen S-construction and Hall algebras.
Classically, since the work of Rota~\cite{RotaMoebius} in the 1960s, incidence algebras
were defined from posets through the process of decomposing
intervals, a construction which can fruitfully be formulated in terms of
the nerve of a poset. However, a great many combinatorial co- and
Hopf algebras are not of incidence type, meaning that they are not
the incidence coalgebra of a poset. The basic observation of
\cite{Galvez-Kock-Tonks151207573} is that simplicial objects more general than nerves of
posets admit the incidence coalgebra construction, allowing most of
the basic features of the theory to carry over --- and that most
combinatorial Hopf algebras do arise as incidence Hopf algebras of
(monoidal) decomposition spaces.

\section{Double-Pushout rewriting and tracelet theory}

\begin{figure}
\centering
\subfigure[A direct derivation sequence of length $5$ (with edge creation/deletion 
rules, and where ``wires'' indicate matches).\label{fig:ddExample}]{
$\vcenter{\hbox{\includegraphics[scale=0.1]{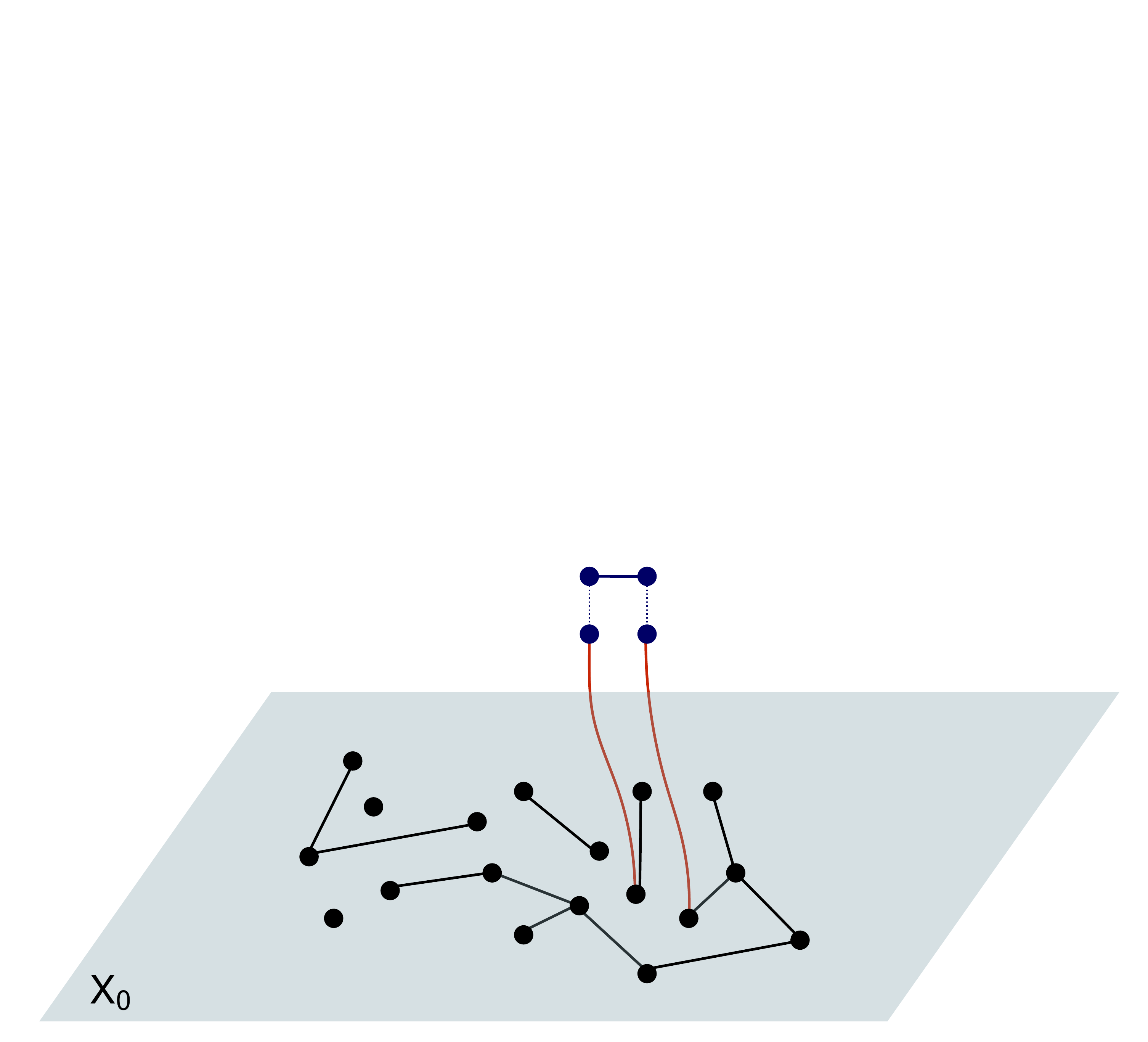}}}$
$\vcenter{\hbox{\includegraphics[scale=0.1]{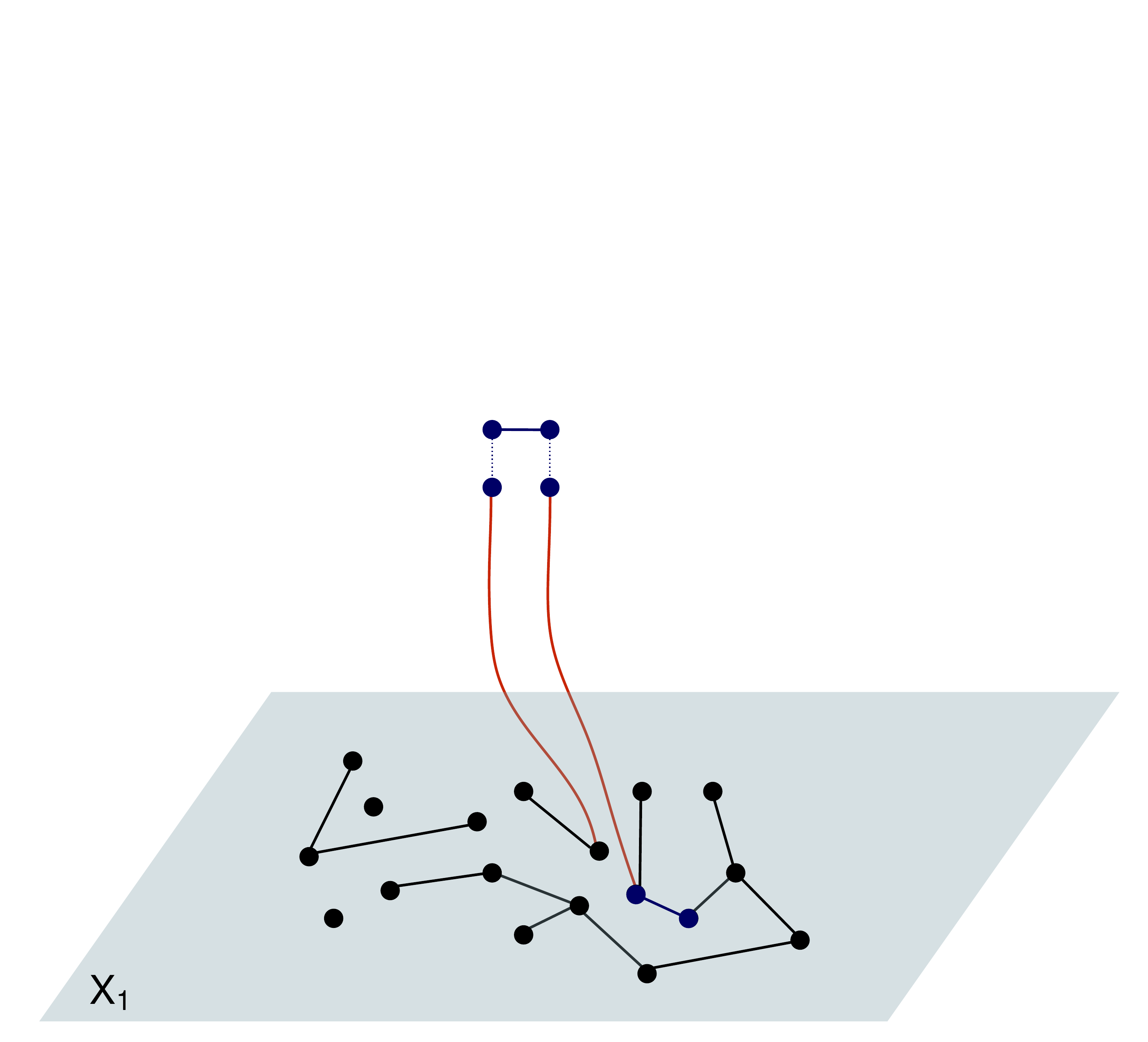}}}$
$\vcenter{\hbox{\includegraphics[scale=0.1]{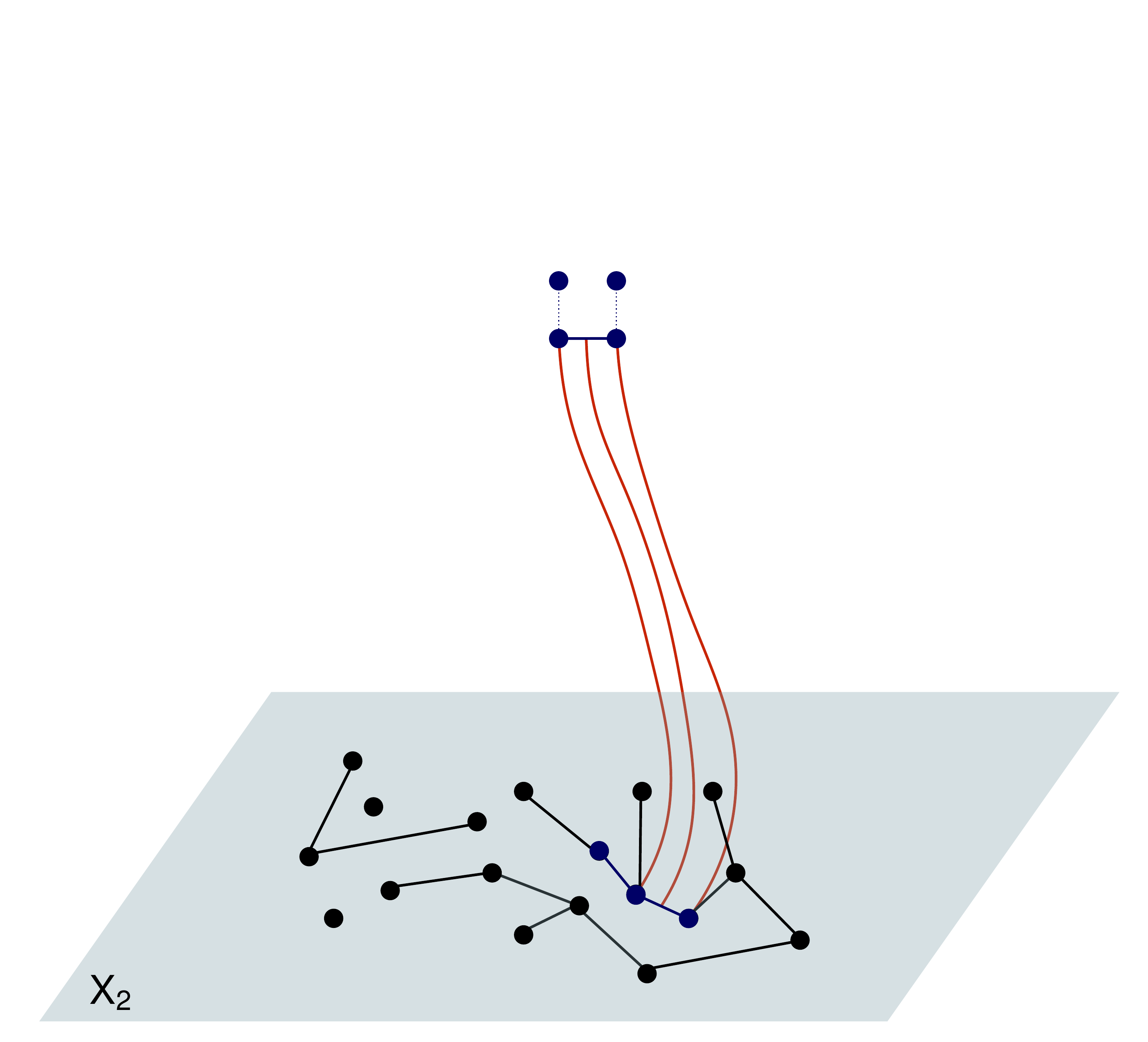}}}$
$\vcenter{\hbox{\includegraphics[scale=0.1]{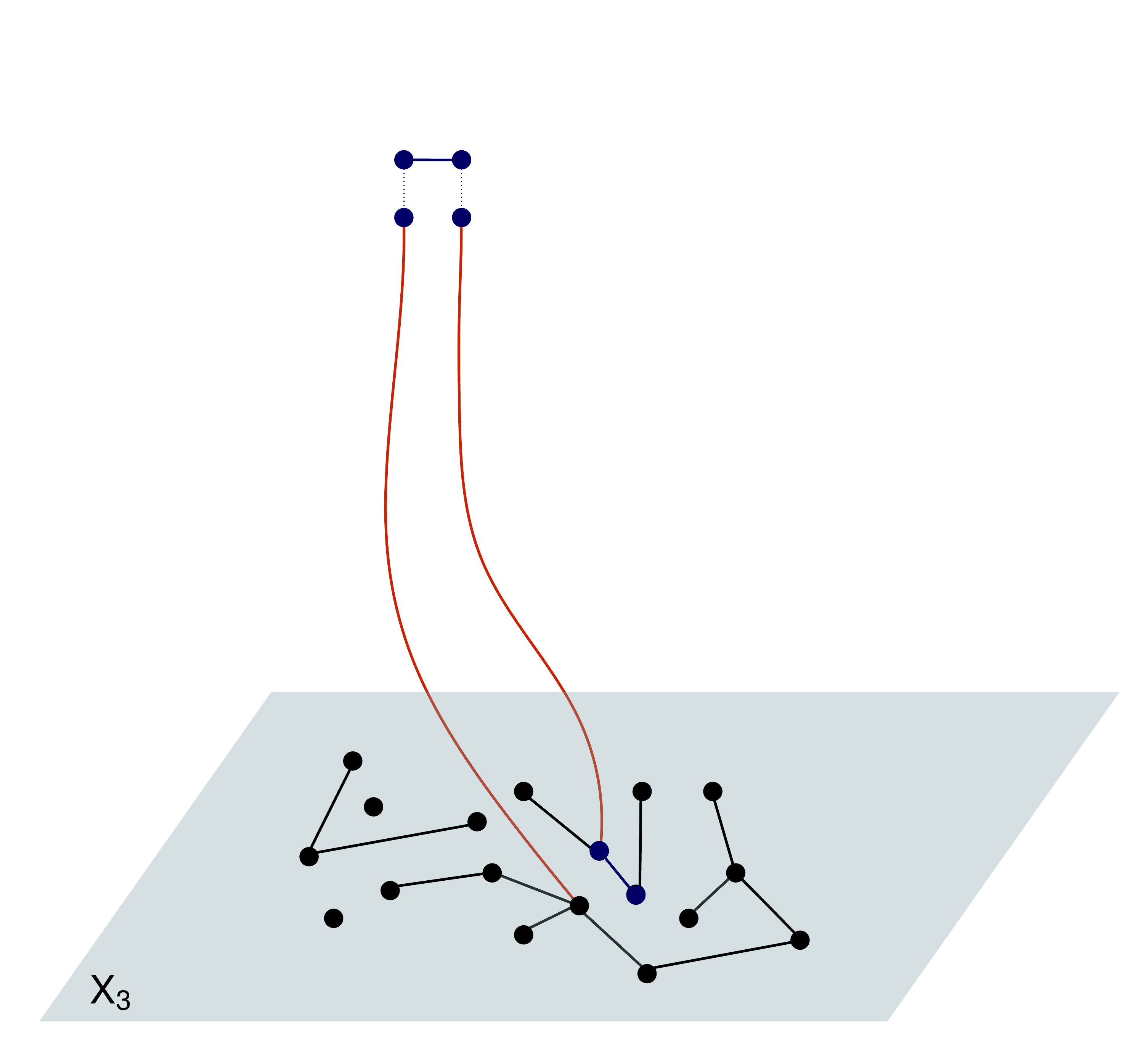}}}$
$\vcenter{\hbox{\includegraphics[scale=0.1]{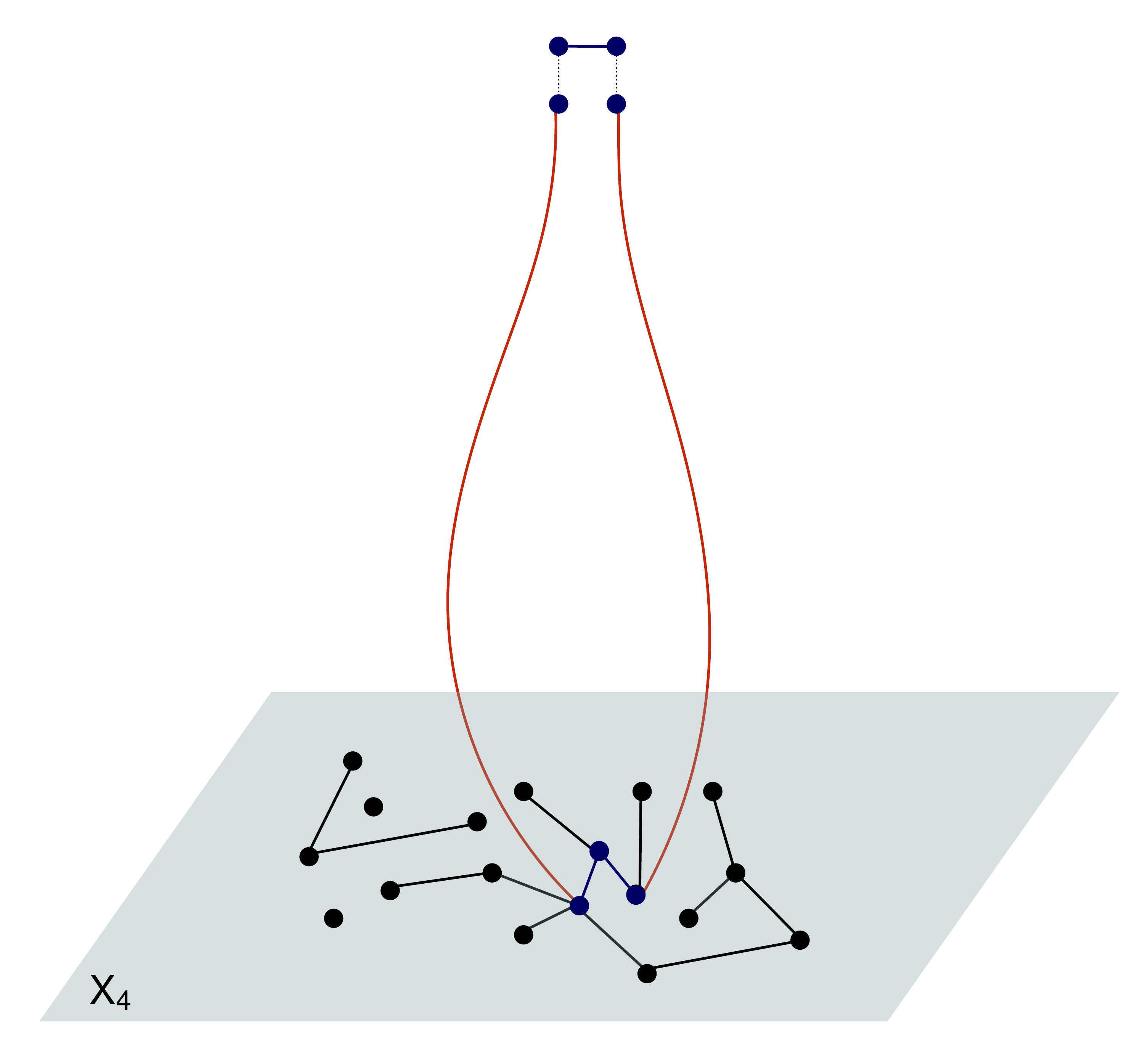}}}$}\\
\subfigure[Tracelet and shift equivalence example.\label{fig:TAillustration}]{
$\vcenter{\hbox{\includegraphics[height=0.32\textwidth]{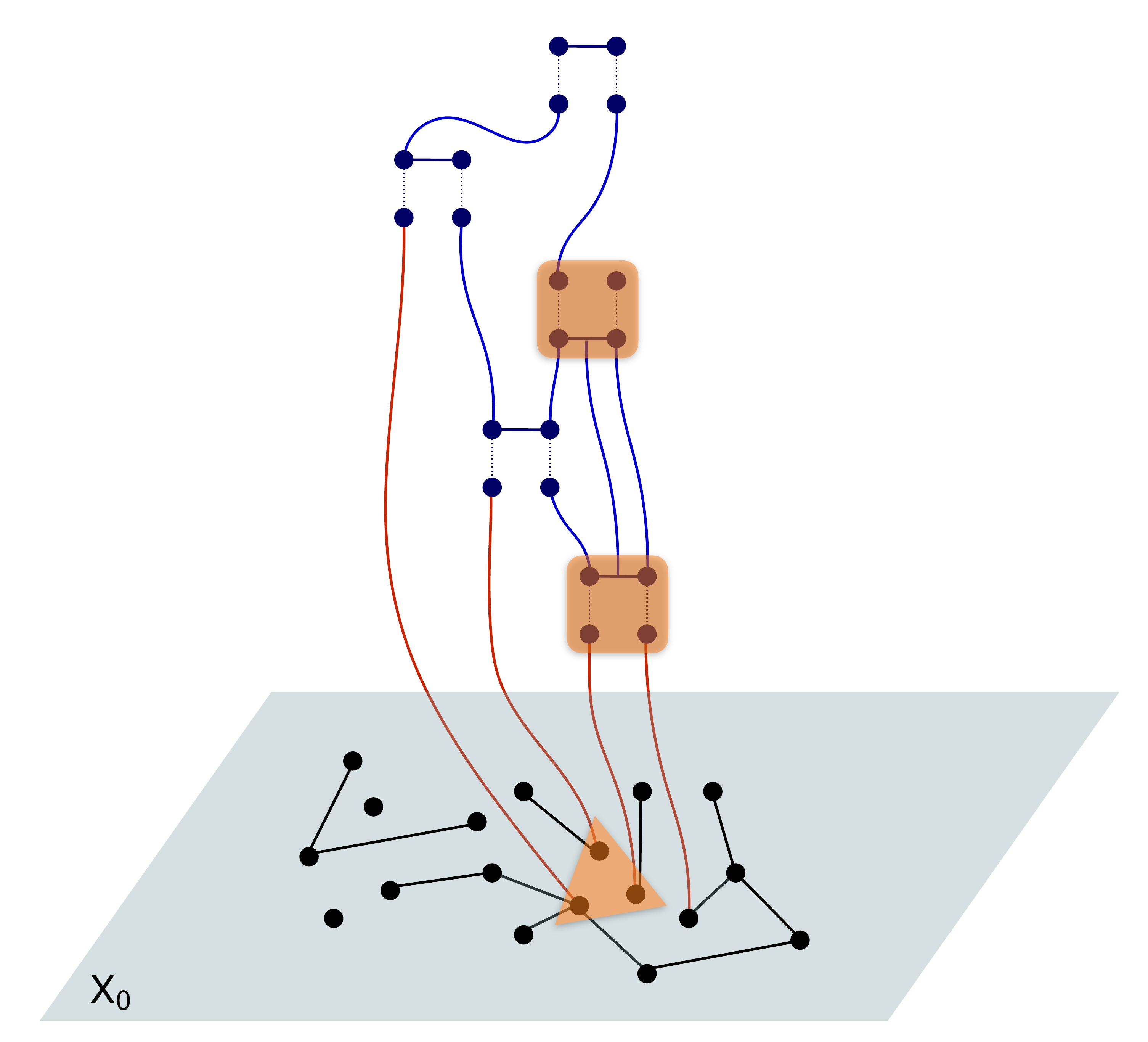}}}$}\;
\subfigure[Defining property of tracelets (here of length $3$).\label{fig:TA3definingProperty}]{$\vcenter{\hbox{\includegraphics[height=0.32\textwidth]{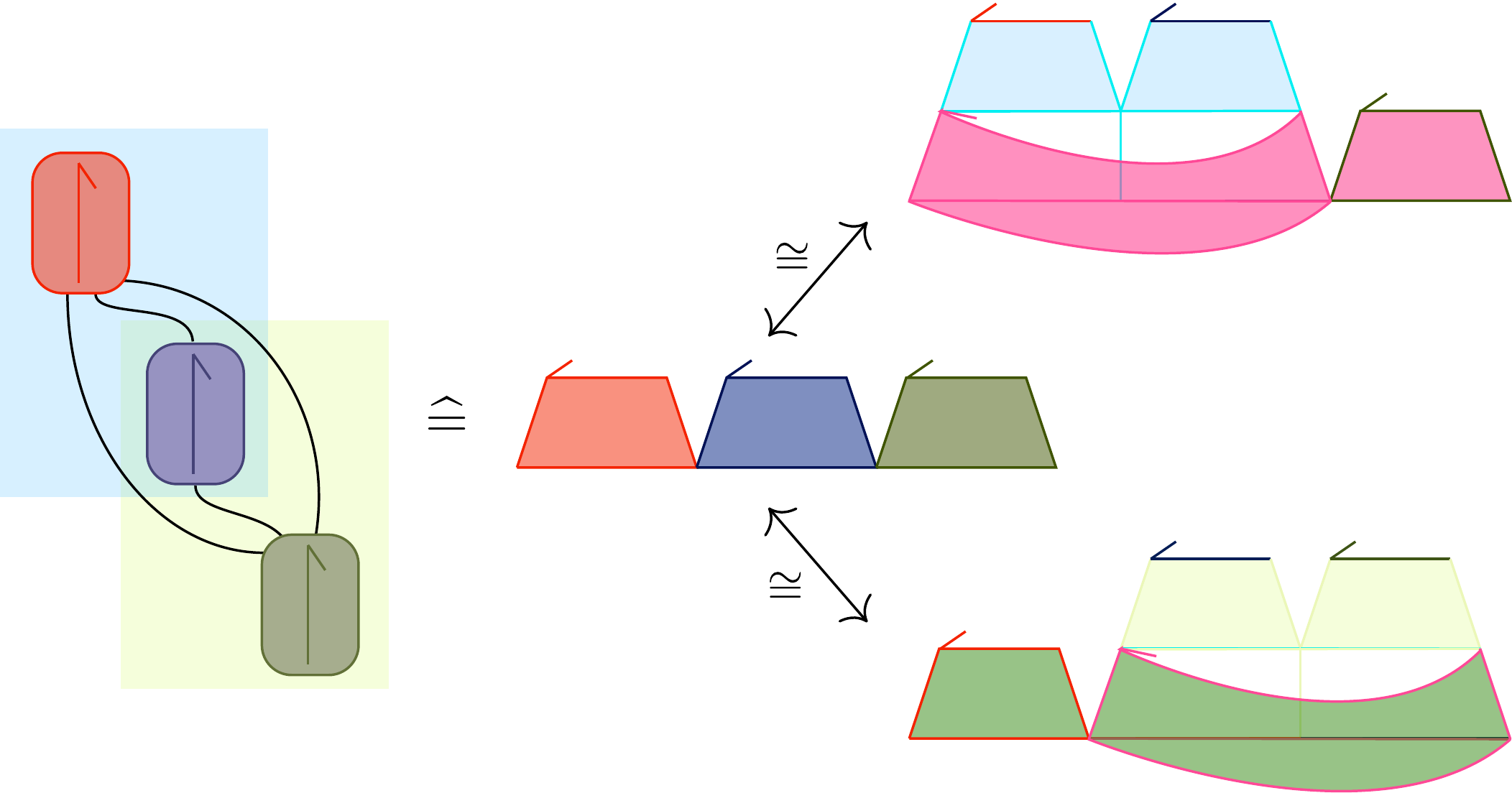}}}$}
\caption{An illustration of graph rewriting sequences (top) and of the tracelet picture (bottom).}
\label{fig:TLP}
\end{figure}

Consider a rewriting system of undirected multi-graphs such as the one depicted 
in Figure~\ref{fig:ddExample}, with some elementary rules that link and unlink vertices 
with edges. Starting from some initial graph $X_0$, each rewriting step (\emph{direct 
derivation}) consists in choosing a rewriting rule together with an occurrence 
(\emph{match}) of the input motif (drawn at the bottom of the rule diagrams) within
the graph that is rewritten. %
Sequences of rewrites (\emph{derivation traces}) have a rich
intrinsic structure, arising from the highly non-trivial \emph{interactions} of rules in a sequence through their matches, rendering a \emph{compositional} interpretation of derivation traces a highly non-trivial task. %
To wit, consider the diagram in Figure~\ref{fig:TAillustration}, which provides a
sort of movie-script depiction of the five-step derivation trace depicted in the top
part of the figure. For clarity, red wires are used to indicate inputs to rules that
are present in the original configuration $X_0$, while blue wires indicate inputs to
rules that have originated from outputs of preceding rules.

The main purpose of the theory of \emph{tracelets}~\cite{behr2019tracelets} then consists in rendering mathematically precise the meaning of this intuitive picture. In particular, according to the \emph{Tracelet Characterization Theorem}~\cite[Thm.~2]{behr2019tracelets}, each derivation trace of length $n$ is uniquely characterized by a \emph{tracelet} of length $n$ (cf.\ the sub-diagram consisting of the five rules and all blue wires in Figure~\ref{fig:TAillustration}) and a \emph{match} of the tracelet into the initial configuration $X_0$ (depicted as red wires in Figure~\ref{fig:TAillustration}). Crucially, the compositional structure of tracelets offers a form of static causal analysis via algebraic relations such as commutator relations. This type of analysis takes advantage of algebraic relations such as \emph{shift equivalence}, which in the example of Figure~\ref{fig:TAillustration} amounts to the observation that the rules in the boxes highlighted in orange may be freely moved ``along the wires'' so as to exchange their order (i.e., without changing the overall effect of the rewriting sequence). Finally, as sketched via the highlighted triangle pattern that is produced via the rewriting sequence in the example, tracelets permit to statically reason about the combinatorics of pattern-counting problems in an efficient manner  (cf.\ \cite{Behr2021} for a prototype of such an analysis in the setting of counting patterns in planar rooted binary trees).

As depicted in Figure~\ref{fig:TA3definingProperty}, a tracelet of 
length $3$ exhibits already quite a non-trivial compositional 
structure, in that as sketched the internal structure of partial 
overlaps of rule inputs and outputs in such a tracelet is 
not of a purely sequential nature; to wit, the diagram encodes a special kind of trace of length $3$, with the defining property that it may be equivalently (up to isomorphisms) be obtained via nested composition operations. It is in this particular sense that tracelets offer a \emph{minimal} causal presentation of the structure of rewriting sequences, since via the equivalences to nested pairwise composition operations, they permit to efficiently express $n$-step sequences as just a special type of derivation sequences.

\begin{blanko}{Categorical rewriting theory.} %
For simplicity, we will focus here on the variant of tracelet theory for so-called \emph{Double-Pushout (DPO) rewriting}, and for rewriting rules
without application conditions (referring to~\cite{behr2019tracelets} for
the general theory\footnote{Available generalizations include the type of
rewriting (with \emph{Sesqui-Pushout} semantics an alternative option), the
choice of base-categories (with \emph{$\cM$-adhesive
categories}~\cite{ehrig2014mathcal} a more general option), as well as the
inclusion of constraints and application conditions into the compositional
rewriting semantics (cf.\ also~\cite{Behr2021compositionality}).}).
Throughout this paper, let $\bfC$ denote an \emph{adhesive 
category}~\cite{ls2004adhesive} assumed to be 
\emph{finitary} (in the sense of~\cite{Braatz:2010aa}, i.e., with 
only finitely many subobjects for each object up to isomorphisms), 
and assumed to possess a strict initial object $\mIO\in \obj{\bfC}$. \emph{Rewriting rules} are 
defined as spans of monomorphisms $r = ( O \leftarrow o- K -i\rightarrow I )$, also denoted for brevity as $r=(O\leftharpoonup I)$. In the tradition of 
rewriting theory, we refer to such rules as \emph{linear rules} 
(with ``linear'' referring to the nature of the span as a span of monos), 
and denote the class of all such rules as $\Lin{\bfC}$. For every $r\in \Lin{\bfC}$ 
and object $X\in\obj{\bfC}$, let $\RMatchGT{}{r}{X}$ denote the 
\emph{set of (DPO-admissible) matches of $r$ into $X$}, where
$m\in \RMatchGT{}{r}{X}$ iff $m$ is a monomorphism and the 
\emph{pushout complement} marked $\mathsf{POC}$ in the left-most diagram exists:
\begin{equation}
\ti{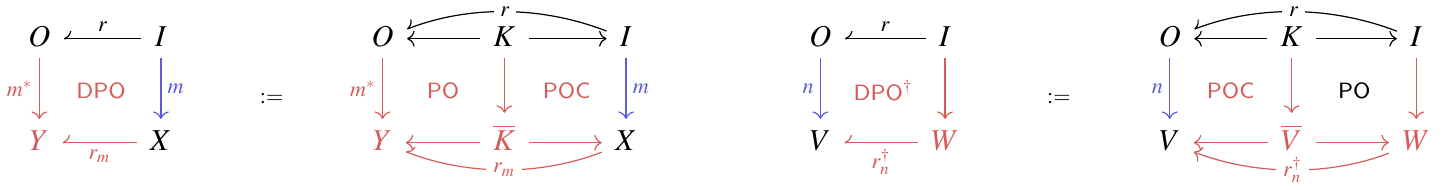}
\end{equation}
Note that in an adhesive category, pushouts along monomorphisms (here marked 
$\mathsf{PO}$) are guaranteed to exist; in contrast, pushout complements (here 
marked $\mathsf{POC}$) may fail to exist, since not every composable pair of arrows 
can be completed into a pushout square. Moreover, $r_m$ and $Y=r_m(X)$ are evidently 
only defined up to universal isomorphisms. It is customary to refer to the data of 
the aforementioned diagram as a \emph{direct derivation}. For later convenience, 
taking advantage of the symmetry of the definition, we mark by $\mathsf{DPO}^{\dag}$ 
diagrams that arise as DPO-type direct derivations in the ``opposite direction'', 
i.e., ``against'' the direction of rules (here, from left to right; cf.\ the 
right-most diagram above).
\end{blanko}

\begin{blanko}{Tracelets.}\label{sec:maindefTracelets}
The class of \emph{tracelets} $\cT$ for DPO-type type rewriting over $\bfC$ is defined recursively:
\begin{itemize}
\item \emph{Tracelets of length $1$}: for every rule $r=(O\leftharpoonup I)\in \Lin{\bfC}$, define $T(r)\in \cT_1$ as a diagram
\begin{equation}
\ti{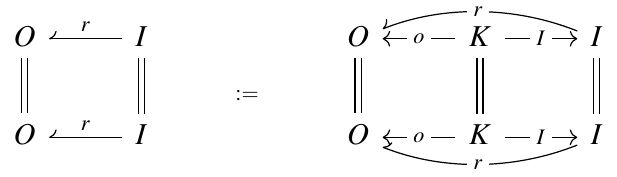}\,.
\end{equation}
\item \emph{Tracelets of length $n+1$}: denoting by $\cT_n$ (for $n\geq1$) the class of trecelets of length $n$, the class $\cT_{n+1}$ is defined to consist of diagrams as below, where $r_k=(O_k\leftharpoonup I_k)\in \Lin{\bfC}$ are linear rules (for $k=1,\dotsc,n+1$), where the top right part of the diagram encodes a tracelet $T\in \cT_n$ of length $n$, where $\mu=(I_{n+1}\leftarrow M\rightarrow O_{n\cdots1})$ is a span of monos, with the cospan $I_{n+1}\rightarrow Y^{(n+1)}_{n+1,n}\leftarrow O_{n\cdots 1}$ its pushout, and such that the direct derivations marked $\mathsf{DPO}$ and $\mathsf{DPO}^{\dag}$ exist:
\begin{equation}
\ti{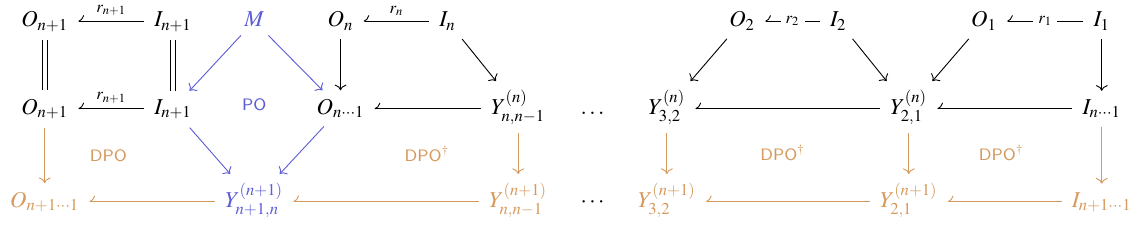}
\end{equation}
Then this data defines a tracelet $\TcompGT{T(r_{n+1})}{\mu}{T}{}$ of length $n+1$ (the \emph{tracelet composition} of $T(r_{n+1})$ with $T$ along $\mu$) uniquely up to universal isomorphisms (see comments below) as
\begin{equation}\label{eq:TnPone}
\TcompGT{T(r_{n+1})}{\mu}{T}{}:=\ti{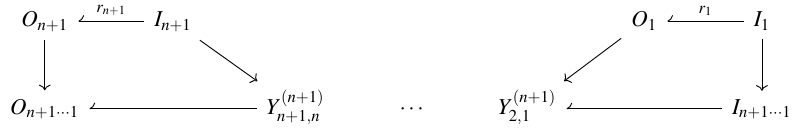}\,.
\end{equation}
\end{itemize}
For later convenience, we will sometimes speak of the commutative 
square adjacent to rule $r_i$ in a tracelet as the \emph{$i$th plaquette}. 
Let $\cT:=\cup_{n\geq1}\cT_n$ denote the class of all (finite-length) tracelets. For later convenience, we also introduce the notations $\mathsf{in}(T):=I_{n\cdots1}$ (``input interface'' of $T\in \cT_n$), $\mathsf{out}(T):=O_{n\cdots1}$ (``output interface'' of $T\in \cT_n$), $\tMatchGT{T(r_{n+1})}{T}{}$ (for ``matches'', i.e., admissible partial overlaps $\mu$ of the length-$1$ tracelet $T(r_{n+1})$ with the tracelet $T\in \cT_n$) and $[[T]]$ for the so-called \emph{evaluation} of the tracelet $T$, which for $T\in \cT_n$ is defined with notations as in the top right of~\eqref{eq:TnPone} as
\begin{equation}
[[T]]:=(O_{n\cdots1}\leftharpoonup I_{n\cdots 1})=(O_{n\cdots1}\leftharpoonup Y^{(n)}_{n,n-1})\circ \cdots\circ (Y^{(n)}_{2,1}\leftharpoonup I_{n\cdots 1})\,.
\end{equation}
Here, $\circ$ denotes the operation of \emph{span composition} (considered up to 
span isomorphisms).
\end{blanko}

Up to this point, one might say that DPO-type tracelets are some form of data structure that encodes a certain form of \emph{sequential compositions} of rewriting rules. This point of view is augmented via the following definition, which finally reveals tracelets as a particular notion of compositional diagrams.

\begin{blanko}{Tracelet composition.} Let $T\in \cT_m$ and $T'\in \cT_n$ be two tracelets of length $m$ and $n$, respectively (for $m,n>0$). Let $\mu:=(I_{m\cdots1}\leftarrow M\rightarrow O'_{n\cdots 1})$ be a partial overlap (of the ``input interface'' $I_{m\cdots1}$ of $T$ with the ``output interface'' $O'_{n\cdots 1}$ of $T'$) whose pushout $I_{m\cdots1}\rightarrow Y^{(m+n)}_{n+1,n}\leftarrow O'_{n\cdots 1}$ satisfies that in the diagram below, all direct derivations marked $\mathsf{DPO}$ and $\mathsf{DPO}^{\dag}$, respectively, exist:
\begin{equation}
\ti{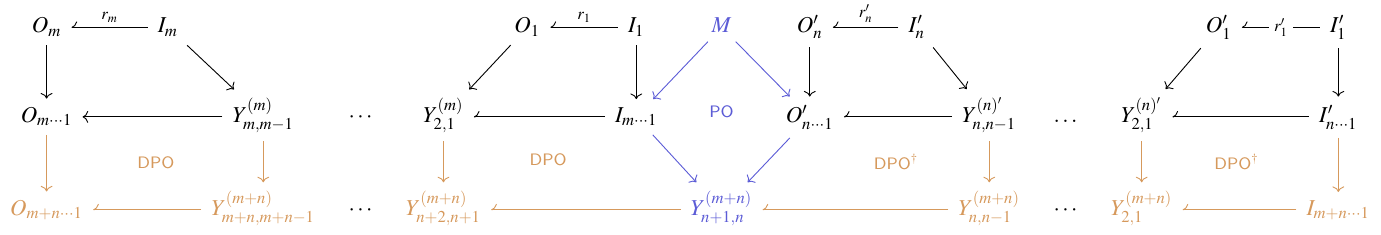}
\end{equation}
In this case, we write $\mu\in \tMatchGT{T}{T'}{}$ to say that $\mu$ is a 
(DPO-admissible) \emph{match} of $T$ into $T'$, and we define the \emph{composition} of $T$ with $T'$ along $\mu$ as
\begin{equation}
\TcompGT{T}{\mu}{T'}{}:=\ti{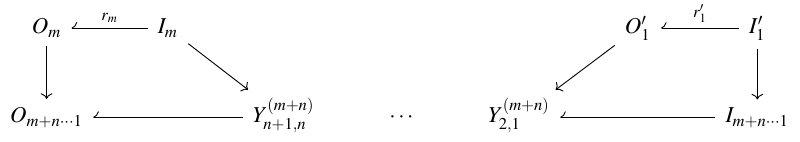}\,.
\end{equation}
\end{blanko}

The definition of tracelets and their composition might appear somewhat ad
hoc at first sight, yet it is very natural if viewed in diagrammatic form.
To this end, consider the example of a tracelet of length $3$ such as in Figure~\ref{fig:TA3definingProperty}. The ``wires'' in the schematic
diagram that link individual length-$1$ tracelets encode the partial
overlaps; as indicated, the tracelet of length $3$ may be realized
recursively by either determining the partial overlap of the first and the
second sub-tracelet, composing, and then determining the resulting overlap
of the composite tracelet of length $2$ with the third tracelet of length
$1$, or (equivalently as it will turn out) by 
computing the composition of the
third and second tracelets of length $1$, and of that composite with the
first tracelet of length $1$. This so-called \emph{associativity property}
of tracelet composition is at the heart of the algebraic properties of
tracelets. It will be further illustrated when we now pass to discuss tracelets
in the framework of decomposition spaces.

\section{The decomposition space of rewrite rules}\label{sec:DSRR}

In this section we describe a decomposition
space $\X_\bullet$ of rewrite rules (for a fixed rewrite system in a fixed
adhesive category $\bfC$ as above), whose incidence algebra is the rule algebra.

\begin{blanko}{Decomposition spaces.}
  A decomposition space~\cite{Dyckerhoff-Kapranov12123563,Galvez-Kock-Tonks151207573} is a simplicial groupoid $X_\bullet :
  \simplexcategory\op\to\Grpd$ satisfying a certain exactness property
  designed precisely to allow the incidence coalgebra construction,
  classically defined for posets. The nerve of a poset or a category is an
  example of a decomposition space. Where categories encode composition,
  decomposition spaces owe their name to encoding more generally decomposition. Many
  situations where compositionality is hard to achieve can be dealt with
  instead with decompositions, as is often the case in combinatorics, where
  combinatorial structures can be split into smaller ones without the
  ability to compose~\cite{Galvez-Kock-Tonks170802570,Galvez-Kock-Tonks161209225}. Often non-deterministic composition
  structures can be turned around and constitute instead a decomposition.

  There are different ways to formulate the decomposition-space axioms. One
  (simplified) version states that for any endpoint-preserving monotone map
  $\alpha: [2] \actto [n]$, defining a decomposition of any $n$-simplex
  into an $n_1$-simplex and an $n_2$-simplex, the natural square
  \begin{equation}
  \ti{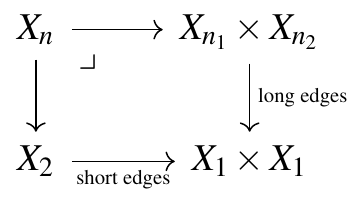}
  \end{equation}
  is a (homotopy) pullback. It says that an $n$-simplex can be reconstructed from the 
  two smaller simplices of the decomposition together with the information of
  a gluing of the long edges of the two simplices onto the short edges 
  of a base $2$-simplex.

  This condition is considerably weaker than the Segal condition (which
  characterizes categories, hence composition rather than just
  decomposition), which says that a single-vertex overlap between the two
  smaller simplices is enough to perform the gluing. In the decomposition-space case, the
  base $2$-simplex is required as a kind of context for the gluing.
\end{blanko}

\begin{blanko}{Groupoids of tracelets.}
  An {\em isomorphism} between two tracelets of length $n$ is by definition
  a family of object-wise isomorphisms between the involved objects in
  $\bfC$ making all squares commute. %
  We denote by $\X_n$ the groupoid of all tracelets of length $n$.
  (In particular, 
  the only tracelet of length $0$ is the empty one (which 
  evaluates to the trivial rule), so $\X_0 = \{*\}$.)
\end{blanko}

\begin{theorem}
  The groupoids $\X_n$ assemble into a simplicial groupoid
  $\X_\bullet: \simplexcategory\op\to\Grpd$ (whose face and degeneracy maps we 
  proceed to describe below).
\end{theorem}

Recall that a simplicial structure amounts to face maps $d_i$ and degeneracy maps 
$s_i$ as in the diagram
\begin{equation}
\ti{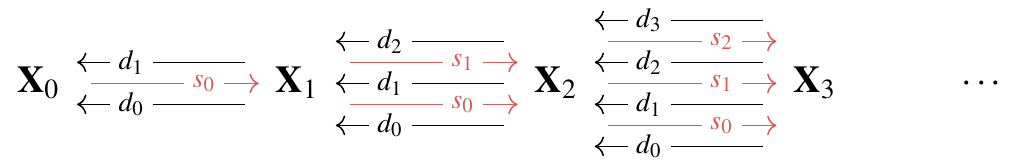}
\end{equation}
subject to the {\em simplicial identities}:
  $d_is_i=d_{i+1}s_i=1$ and
  $$
  d_id_j=d_{j-1}d_i,\quad
  d_{j+1}s_i=s_id_j,\quad
  d_is_j=s_{j-1}d_i,\quad
  s_js_i=s_is_{j-1}\quad
  \qquad(i<j).
  $$
  The bottom and top face maps generate the class of {\em inert maps}, whereas the 
  inner face maps and the degeneracy maps generate the class of {\em active maps}, 
  for which we use the special arrow symbol $\actto$. These two classes of maps play 
  a special role in the theory; see~\cite{Galvez-Kock-Tonks151207573}, where more 
  conceptual characterizations are given.

\begin{blanko}{Description of the face maps.}
The \emph{top face map} $d_n:\X_n\to\X_{n-1}$ (resp.~the \emph{bottom face
  map} $d_0:\X_n\to\X_{n-1}$) is defined via (1) performing \emph{tracelet surgery} 
  to exhibit the tracelet as a composition of the first (the last) rule with an ($n-1$)-tracelet, followed by (2) extracting the length-($n-1$) tracelet. This is illustrated in Figure~\ref{fig:TA3definingProperty} for the case of $n=3$, with $d_0$ ($d_3$) defined to return the tracelet shaded in light blue (in light yellow). 

Recall from Section~\ref{sec:maindefTracelets} that a \emph{plaquette in position $i$} of a given tracelet is defined as the $i$th direct derivation, i.e., the commutative subdiagram of the tracelet involving the $i$th rule (read from the right). Then the \emph{inner face maps} $d_i : \X_n \to \X_{n-1}$ (for $0<i<n$) replace the two plaquettes $p_i$ and $p_{i+1}$ in the chain with a single new plaquette $p'$ having the same starting point as $p_i$ and the same endpoint
  as $p_{i+1}$, by applying the ``synthesis'' part of the concurrency theorem  to convert the sub-sequence of plaquettes $p_{i+1}$ after $p_i$ into a
  one-step direct derivation along the composite rule $p'_{i+1,i}$. The
  result of this operation is guaranteed to be a tracelet of length $n-1$. Referring 
  once again to Figure~\ref{fig:TA3definingProperty} for an illustration of the case $n=3$, $d_1$ ($d_2$) are defined to return the tracelet shaded in green (in pink).

  The \emph{degeneracy maps} $s_i : \X_n \to \X_{n+1}$ (for $0\leq i \leq n$)
  insert a copy of the trivial rule $\emptyset \leftarrow \emptyset \to 
  \emptyset$ in the tracelet at position $i$.
\end{blanko}

In the form stated, $\X_\bullet$ is only a pseudo-simplicial groupoid.
This means that the simplicial identities only hold up to (specified)
isomorphism, and that there are coherence issues to deal with. The reason
for this pseudo-ness is that composition of rules and tracelets, as
involved in the face maps, is only well defined up to isomorphism, relying
as it does on pushouts and pullbacks. To actually get well-defined face maps, it is
necessary to make choices of these universal constructions, and these
choices screw up the strict simplicial identities. (A well-known example 
of this phenomenon is
how composition of spans by means of pullbacks defines a bicategory, not an
ordinary category.)

This pseudo-ness is not at all a problem for the sake of decomposition-space 
theory, designed to be up to homotopy, and it does not affect the incidence
algebra we construct from this decomposition space (which in any case is 
spanned by iso-classes of rewrite rules). Nevertheless it is very fruitful
to provide also a strict model of $\X_\bullet$. The standard technique for 
constructing this (which goes back to insight from algebraic topology from 
the 1970s (notably Quillen,\footnote{Historical remark: B\'enabou (1963) 
had described a bicategory spans. Quillen used the techniques of big 
redundant $n$-simplices to exhibit the same structure as a {\em strict} 
simplicial groupoid, now called Quillen's Q-construction. Instead of 
having simply chains of $n$ composable spans 
\raisebox{-0.5mm}{\ti{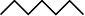}}
in degree $n$, he defined it 
to be diagrams of shape \raisebox{0.5mm}{%
\ti{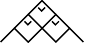}}, that is composable spans, {\em together} with all 
the relevant pullbacks. Similar constructions were given in related 
situation by Waldhausen and Segal, and today the technique is standard in 
algebraic topology.} Waldhausen, and Segal)) is to beef up the 
groupoid of $n$-simplices to something equivalent that contains all the
(redundant) data involved in the face maps.

Specifically, a $2$-simplex should not just be a $2$-tracelet, but rather 
a $2$-tracelet {\em together} with a choice of composite rule. In this way 
the middle face map $d_1 : \X_2 \to \X_1$ does not have to compute any composite by means of 
choices; it can simply return the choice already built in. The fact that
these choices are unique up to universal isomorphisms says precisely that
this bigger groupoid is equivalent to the original, and hence that the
homotopy properties of the bigger simplicial groupoids are the same.
In Figure~\ref{fig:r2s} we see such a fully specified $2$-simplex.
The two short edges ($01$ and $12$) are the two rules in a $2$-tracelet,
and the squares marked PO and POC are the plaquettes constituting 
altogether the $2$-tracelet. The pullback square (blue, marked with PB) is 
not
part of the data of the tracelet, but it is included in the fully 
specified notion of $2$-simplex.

In degree $3$ we arrive at the first point where there is an interesting
simplicial identity to establish, namely commutativity of the square
\begin{equation}
\ti{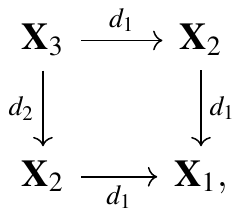}
\end{equation}
which in essence states that a
sequential composition of three rules may be recovered equivalently from
two steps of pairwise rule compositions in either of the nesting orders. 

For the groupoids of bare tracelets, this simplicial identity cannot be
strict, due to the choices of pushouts and pullbacks involved in
composition of rules and tracelets. That the equation holds up to natural
isomorphism is a nontrivial statement which involves the concurrency
theorem (in the particular form called associativity
theorem~\cite{bp2019-ext,bdg2016,bdgh2016}). We explain how the same
theorem implies the strict equation for the fully specified $3$-simplices.
This exhibits the beautiful geometry inherent in the associativity theorem.
As always, the idea is that a fully specified $3$-simplex should contain
all information about all choices. In particular (in order for the four
face maps to be forgetful) it should contain four $2$-simplices of the form
of Figure~\ref{fig:r2s}. A full picture of such a subdivided tetrahedron
is given in Figure~\ref{fig:r3s}.
One can chase through how this is built up from composition of tracelets, 
over specified overlaps:
Consider the diagram
depicted in Figure~\ref{fig:rds3-prepA}, which is formed by (1) a
$2$-simplex encoding a composition of two rules $r_{21}$ and $r_{10}$ into
some rule $r_{20}$, and (2) another $2$-simplex of which one ``short edge''
is the rule $r_{20}$, and which contains another rule $r_{32}$ and the data
of the composition of $r_{32}$ with $r_{20}$ into some rule $r_{30}$. Upon
closer inspection, it is possible (via a number of somewhat intricate
steps) to construct from this data the interior and the other two faces of
a tetrahedron. To this end, one first invokes the ``analysis'' part of the
DPO-type concurrency theorem in order to obtain, from the sub-diagram that
encodes the one-step direct derivation of the object $I_{03}$ along the
composite rule $O_{20}\leftarrow K_{20}\rightarrow I_{20}$, the data of a
sequence of two direct derivations along the ``constituent'' rules
$O_{21}\leftarrow K_{21}\rightarrow I_{21}$ after $O_{10}\leftarrow
K_{10}\rightarrow I_{10}$. This construction in particular delivers an
object $Z$ located in the interior of the tetrahedron. Over several
further steps (involving pushout and pullback operations), it is then
possible to fill the remaining two faces of the $3$-simplex with the
structure of two sequential rule compositions, ultimately resulting in the
diagram of Figure~\ref{fig:r3s}. The fact that all these constructions are
given by universal properties (pushouts and pullbacks, together with
the axioms of adhesive categories) ensures that the groupoid of such
fully specified $3$-simplices is equivalent to the groupoid of bare
$3$-tracelets.
The face maps are now obvious (or even tautological) and all the 
simplicial identities are clearly strict for this reason: they merely 
return data already contained in (the beefed-up version of) $\X_3$.

The higher simplices are
increasingly cumbersome to describe, due to our limited 
vision of geometry in dimension higher than $3$, but the principle is easy to
follow: just include all information about all possible composites,
and the overall geometric shape is always a geometric $n$-simplex whose
edges are rules, whose $2$-dimensional faces are as in Figure~\ref{fig:r2s}
and whose $3$-dimensional faces are as in Figure~\ref{fig:r3s}.

The fact that in each dimension the bare tracelets contain information 
necessary and sufficient to reconstruct the full specified simplex is an expression of the central result
of~\cite{behr2019tracelets} that it is indeed tracelets that provide the
minimal carriers of causal information in sequential rule compositions.

We proceed to establish that $\X_\bullet$ is a decomposition space. Since this is 
a homotopy invariant property, we may work with the simple version of 
groupoids of $n$-tracelets. Before the check, let us just note that 
$\X_\bullet$ is 
not a Segal space (a category), because of the non-deterministic nature of 
composition. Specifically, a $2$-simplex
cannot be reconstructed from knowing its two short edges.

\begin{theorem}
  $\X_\bullet$ is a decomposition space. This means that 
  for all $0<i<n$ the
  two squares
\begin{equation}
\ti{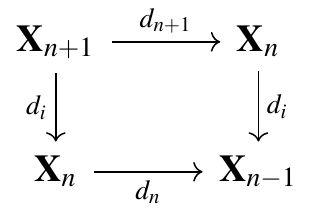}
	\qquad\qquad
\ti{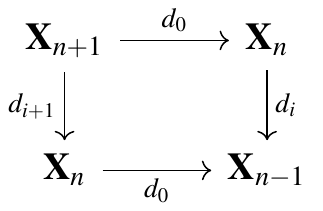}
\end{equation}
are (homotopy) pullbacks.
\end{theorem}

To check this, it is enough to show that the fibers of the maps pictured
vertically are equivalent. We shall see that indeed all fibers of inner
face maps are canonically identified with the fiber of $d_1: \X_2 \to
\X_1$.

\begin{blanko}{Fiber calculations.}
Consider $d_1 : \X_2 \to \X_1$ which sends a pair of composable rules with
minimal gluing $(r_2,w,r_1)$ to the composite rule $r'$. The fiber over $r'\in
\X_1$ is thus the groupoid of all $(r_2,w,r_1)$ that compose to $r'$. We
denote this groupoid $(\X_2)_{r'}$. Notice that the objects of $\X_2$ are composable 
pairs of plaquettes with the property that the intermediate point between the two 
plaquettes is a minimal gluing (of the output of rule $r_1$ with the input of rule 
$r_2$; in other words, the middle cospan in the two-step direct derivation sequence is a pushout of its own pullback). 

\end{blanko}

\begin{lemma}\label{lem:hf}
The (homotopy) fiber of $d_i : \X_n 
\to \X_{n-1}$ (for $0<i<n$) over a tracelet which in position $i$ has a plaquette with rule $r'$ is equivalent to the groupoid
$(\X_2)_{r'}$. In particular, it does not depend on the whole plaquette $p'$
under $r'$, and it does not depend on the context in any way.
\end{lemma}

One can now unpack the general construction of incidence algebras of 
decomposition spaces (cf.~\cite{Galvez-Kock-Tonks151207573}) to establish:
\begin{proposition}
  The incidence algebra is the rule algebra of \cite{bp2019-ext}.
\end{proposition}
This algebra is not our main focus in this work. Rather do we regard the 
decomposition space $\X_\bullet$ as a stepping stone towards more interesting 
decomposition spaces and Hopf algebras, notably the tracelet Hopf algebra.

\section{Decomposition spaces of tracelets}\label{sec:DST}

So far we have defined the decomposition space $\X_\bullet$ of rules, whose
incidence algebra is the rule algebra of \cite{bp2019-ext}. We
now proceed towards Hopf algebras spanned by tracelets.

The Hopf algebra of tracelets should be spanned by iso-classes of 
tracelets, which are now furthermore required to be {\em non-degenerate} as simplices of $\X_\bullet$. This means that the rules
involved are not allowed to be the trivial rule.\footnote{By imposing this condition,
we account directly for an equivalence relation imposed in \cite{behr2019tracelets}
called `equivalence up to trivial tracelets' (cf.\ Definition~\ref{def:nfe}).} The
non-degenerate simplices of $\X_\bullet$ do not form a simplicial object, since inner
faces of non-degenerate simplices are not always non-degenerate, but the outer face
maps (the inert maps) survive (as a consequence of the 
decomposition-space axioms, see \cite{Galvez-Kock-Tonks151207577}), so as to define a
presheaf
$$
\nondeg \X_\bullet: \simplexcategory_{\text{inert}}\op\to\Grpd .
$$
Left Kan extension along 
the inclusion functor $j: \simplexcategory_{\text{inert}} \to \simplexcategory$
defines a new simplicial groupoid:
$$
\Y_\bullet := j\lowershriek \nondeg \X_\bullet \ : \ \simplexcategory\op\to\Grpd . 
$$
This is a general construction that makes sense for any (complete)
decomposition space, and by a result of Hackney and
Kock~\cite{Hackney-Kock} it always produces a decomposition space again.
One can expand explicitly what its simplices are:
$$
\Y_k = \sum_{\alpha:[k] \actto [n]} \nondeg \X_n .
$$
(The sum is over active maps.)
In particular
$$
\Y_0 = \X_0 \qquad \text{and} \qquad \Y_1 = \sum_{n\in \N} \nondeg \X_n .
$$
So the new $1$-simplices are the non-degenerate tracelets of any length.
The higher simplices are `subdivided tracelets'. To see this, recall that 
the decomposition space axioms can be written  (cf.~\cite[Prop.~6.9]{Galvez-Kock-Tonks151207573}) as saying that for any 
active map $\alpha:[k] \actto [n]$ the canonical square
\begin{equation}
\ti{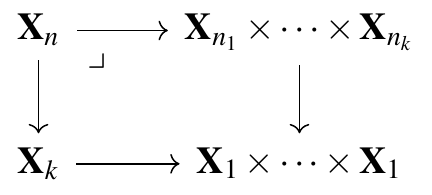}
\end{equation}
is a (homotopy) pullback.
Here the vertical maps are active and the horizontal maps are combinations 
of inert maps.  What the condition says is that it is possible to glue 
together $k$ simplices (of different dimensions $n_i$) if just one has 
available a `mould' to glue them together in, namely a $k$-simplex whose 
$k$ principal edges match the long edges of the $k$ simplices.
(This is also the essence of the very definition of tracelet.) 

An example of such a composition is depicted in Figure~\ref{fig:rds3-prepA}, in 
which a length-$2$ tracelet (depicted as the $2$-simplex 012) is composed along the short edge 02 of the $2$-simplex 023 with a tracelet of length $1$ (here depicted as the edge 23). Figure~\ref{fig:rds3-prepB} then depicts the method for computing the resulting tracelet of length $3$, which itself is depicted in Figure~\ref{fig:rTot}.

Since non-degeneracy in a decomposition space can be measured on principal 
edges (cf.\ \cite{Galvez-Kock-Tonks151207577}), we also have the 
(homotopy) pullback
\begin{equation}
\ti{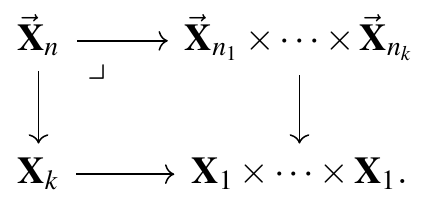}
\end{equation}
We see that a $k$-simplex in $\Y_\bullet$ is the data of a tracelet $\tau$ of length $k$
(not necessarily non-degenerate) together with a non-degenerate tracelet 
$\sigma_i$
glued onto each of the principal edges of this base tracelet along their
evaluation. (That is, the rule given by evaluating the tracelet $\sigma_i$ must match
the rule corresponding to the $i$th principal edge of $\tau$.)

The corresponding algebra, given by the standard incidence algebra
construction (cf.~\cite{Galvez-Kock-Tonks151207573}), is spanned by
isomorphism classes of non-degenerate tracelets, and the product of
two tracelets is given by summing over all possible tracelet
composites.

\bigskip

We now proceed to extend this structure into a Hopf algebra.
This is not straightforward, because the decomposition space 
$\Y_\bullet$ is not monoidal under sum.  A monoidal structure exists in degree $1$, by 
declaring the product of two tracelets to be the composite along trivial 
overlap:
$$
T \odot T' :=  
\TcompGT{T}{\emptyset}{T'}{} .
$$
But this definition is not compatible with higher simplices.

Our task is now to explain how this is fixed in a canonical 
way.
The solution amounts to imposing the so-called {\em shift equivalence} 
relation on tracelets, an equivalence relation already important in rewriting 
theory. In the graphical interpretation it is about saying that for tracelets that are not 
connected, it should make no difference in which order they are applied.
After passing to this equivalence relation, the monoidal structure $\odot$ 
will be well defined in all simplicial degrees. This final symmetric 
monoidal decomposition space of tracelets up to shift equivalence will be 
denoted $\Z_\bullet$. We shall go deeper into the notion of shift equivalence in 
Section~\ref{sec:Hopf} (and interested readers are referred
to~\cite{bp2018,nbSqPO2019,Behr2021compositionality} 
for the full background information and details).
Here we just state the following Proposition~\ref{prop:shiftequiv}, which gives an alternative
approach to shift equivalence.

A {\em splitting vertex} of a tracelet is an inner vertex for which the corresponding rule 
overlap is trivial.
This property is invariant 
under precomposition with active maps. (That is, if $\sigma' = g(\sigma)$
for $g$ an active map not eliminating vertex $v$, then $v$ is splitting for
$\sigma'$ if and only if it is splitting for $\sigma$.) Second, there is a
transitive property in connection with `stages' in the sense of
higher-order simplices of $\Y_\bullet$.
Note that this transitive property does {\em not} imply that 
irreducibility is compatible with inert maps (outer face maps).

A non-degenerate tracelet $T \in \Y_1 = \sum_n \nondeg \X_n$ is {\em
primitive} if it does not admit any splitting. (A higher-dimensional 
simplex $\sigma\in\Y_k$ (that is a subdivided tracelet) is primitive if its long edge is
primitive in $\Y_1 = \sum_n \nondeg \X_n$ (that is, its underlying tracelet
is primitive).)

\begin{lemma}
Every maximal splitting of a given simplex has, up to isomorphism and permutation, the 
same primitive pieces.
\end{lemma}

\begin{proposition} \label{prop:shiftequiv}
  Tracelets are shift equivalent in the restricted sense of trivial overlaps if and only if they have the same factorization into primitives.
\end{proposition}

\begin{proposition}
  Shift equivalence is compatible with the simplicial structure. This
  defines a simplicial groupoid $\Z_\bullet$ with $\Z_k = \Y_k/\sim$. This
  simplicial groupoid is a (locally finite) decomposition space.
\end{proposition}

Note that if $\widetilde \Y_n$ denotes the groupoid of shift equivalence 
classes, then we have
$$
\Z_k = \sum_{[k] \actto [n]} \Y_k \times_{\X_1^k} \big(\widetilde{\Y}_{n_1} 
\times\cdots\times \widetilde{\Y}_{n_k}
\big)
$$
This makes sense: in the fiber product, the maps from the factors 
$\widetilde{\Y}_{n_i}$ return the long edge, which is invariant under 
shift equivalence.

\begin{theorem} There is a level-wise equivalence of groupoids
$$
\Z_k \simeq \SSS(\Y_k^{\operatorname{irr}}) ,
$$
assembling into an equivalence of simplicial groupoids.
here $\SSS$ is the free-symmetric-monoidal-category monad.
In particular, $\Z_\bullet$ is symmetric monoidal under $\odot$.
\end{theorem}

Note that the primitive tracelets themselves do not form a simplicial
groupoid, as the outer face map
applied to a primitive tracelet is not necessarily primitive.
But after we apply $\SSS$, which is just a fancy way of
saying `monomials of' or `families of', it does work.

The upshot is now that the standard incidence algebra construction 
(cf.~\cite{Galvez-Kock-Tonks151207573}) yields a Hopf algebra $H$ of tracelets up
to shift equivalence. %
This is the Hopf algebra we are really interested in, and towards which
the previous ones were preliminary constructions. By 
Poincar\'e--Birkhoff--Witt, $H$ is the enveloping algebra of the Lie 
algebra of primitive tracelets. In the next section we spell out the structure maps of this Hopf algebra 
in details.

\section{The Hopf algebra of tracelets}
\label{sec:Hopf}

The construction given of the tracelet Hopf algebra from the viewpoint of
decomposition spaces gives it a certain canonical feel, but it requires a
lot of machinery. However, the Hopf algebra can also be described directly
(via an extension of the rule diagram Hopf algebra construction of~\cite{bdgh2016}, which was based upon relational calculus), which we briefly describe in
this final section. %
Throughout, we fix a field $\bK$ that will typically be chosen
as either $\bR$ or $\bC$ (or, possibly, $\bQ$). An essential prerequisite
for our Hopf algebra construction is given by the following equivalence relations.
\begin{blanko}{Shift equivalence (cf.\ \cite{behr2019tracelets}).}\label{def:se} %
Let $\equiv_S$ denote the equivalence relation on $\cT$ defined as the
  reflexive symmetric transitive closure of the relation on pairwise
  composition operations on tracelets: let $T= \TcompGT{T_B}{\mu}{T_A}{}$
  (for some admissible match $\mu=(I_B\leftarrow M\rightarrow O_A)$), and
  denote by $[[T_B]]=(O_B\leftarrow K_B\rightarrow I_B)$ and
  $[[T_A]]=(O_A\leftarrow K_A\rightarrow I_A)$ the evaluations of $T_B$ and
  $T_A$, respectively. Suppose $[[T_B]]$ and $[[T_A]]$ are
  \emph{sequentially independent} in the composition along $\mu$, which
  entails that $M$ is isomorphic to both the pullbacks of the cospans
  $K_B\rightarrow I_B\leftarrow M$ and $M\rightarrow O_A\leftarrow K_A$,
  respectively. In this situation we define the composite tracelet $\overline{T}=
  \TcompGT{T_A}{\overline{\mu}}{T_B}{}$ (for $\overline{\mu}=I_A\leftarrow
  M\rightarrow O_B$) to be \emph{shift equivalent} to the tracelet $T=
  \TcompGT{T_B}{\mu}{T_A}{}$.
\end{blanko} 

\begin{blanko}{Normal form equivalence (cf.\ \cite{behr2019tracelets}).}\label{def:nfe} %
Let $\equiv_A$ denote an equivalence relation on $\cT$ (so-called
  \emph{abstraction equivalence}) whereby $T\equiv_A T'$ if $T$ and $T'$
  are tracelets of the same length, and if moreover there exists an
  isomorphism $T\xrightarrow{\cong}T'$ (induced from isomorphisms on
  objects so that the resulting diagram commutes). Let $\equiv_T$ be
  defined as the reflexive symmetric transitive closure of a relation
  whereby for any $T\in \cT$, we let $T\equiv_T T\uplus T_{\mIO}\equiv_T
  T_{\mIO}\uplus T$ (with $T_{\mIO}:=T(\mIO\leftarrow\mIO\rightarrow
  \mIO)\in \cT_1$, and where $\uplus:=\TcompGT{}{\mu_{\mIO}}{}{}$ denotes tracelet composition along trivial overlap). Then we define the \emph{tracelet normal form}
  equivalence relation as
  $\equiv_N:=\mathit{rst}({\equiv_A}\cup{\equiv_T}\cup{\equiv_S})$, i.e., as the reflexive symmetric transitive closure of the union of the aforementioned three relations.
\end{blanko}

\begin{definition}[Primitive tracelets]
  Denote by $\mathfrak{Prim}(\cT_N)$ the set of \emph{primitive tracelets}, defined as
  \begin{equation}
    \mathfrak{Prim}(\cT_N):=\{ [T]_{\equiv_N}\vert T\neq T_{\mIO}\land \not \exists T_A,T_B\neq T_{\mIO}: T\equiv_N T_A\uplus T_B\}\,.
  \end{equation}
\end{definition}
Primitive tracelets play a central role in our construction, since they are in a certain sense the smallest ``indecomposable'' building blocks of tracelets with respect to (de-)composition (just as primitive \emph{rule diagrams} in~\cite{bdgh2016}).

\begin{proposition}[Tracelet normal form]\label{prop:TNF}
  Every tracelet $T\in\cT$ is $\equiv_N$-equivalent to a \emph{tracelet 
  normal form} in the sense that $T_{\mIO}\equiv_N T_{\mIO}$, and\footnote{We chose 
  to make the case distinction explicit in order to emphasize that the normal form of a non-trivial tracelet $T\neq T_{\mIO}$ does itself not contain trivial sub-tracelets, so that manifestly $T_i\in \mathfrak{Prim}(\cT_N)$ in $T\equiv_N \uplus_{i\in I}T_i$. This is clearly the case, since invoking $\equiv_T$ on $\uplus_{i\in I}T_i$ would in effect remove any trivial constituent $T_i=T_{\mIO}$.} $\forall T\neq T_{\mIO}: T\equiv_N \biguplus_{i\in I}T_i$, where $T_i\in \mathfrak{Prim}(\cT_N)$ for all $i\in I$, and with $I$ a (finite) index set.
\end{proposition}

\begin{definition}[Tracelet $\bK$-vector space $\hat{\cT}$]
  Let $\hat{\cT}$ be the $\bK$-vector space spanned by a basis indexed by
  $\equiv_N$-equivalence classes, in the sense that there exists an
  isomorphism $\delta:\cT_N \isopil \mathit{basis}(\hat{\cT})$ from the
  set\footnote{Here, we tacitly assume that the $\equiv_N$-equivalence
  classes indeed form a proper \emph{set}, which is in all known
  applications the case since abstraction equivalence $\equiv_A$ is part of
  the definition of $\equiv_N$. For example, it is well known that
  the isomorphism classes of finite directed multigraphs indeed form a set.} of
  $\equiv_N$-equivalence classes of tracelets $\cT_N:=\cT\diagup_{\equiv_N}$
  to the set of basis vectors $\mathit{basis}(\hat{\cT})$. We will use the
  notation $\hat{T}:=\delta(T)$ for the basis vector associated to some
  class $T\in \cT_N$. We denote by $\mathit{Prim}(\hat{\cT})\subset\hat{\cT}$ the sub-vector space of $\hat{\cT}$ spanned by basis vectors indexed by primitive tracelets.
\end{definition}

\begin{definition}[Tracelet algebra product and unit]
Let $\otimes\equiv \otimes_{\bK}$ be the tensor product operation on the 
$\bK$-vector space $\hat{\cT}$. Then the \emph{multiplication map} $\mu$ and the 
\emph{unit map} $\eta$ are defined via their action on basis vectors of $\hat{\cT}$ as follows:
  \begin{align}
  \begin{split}
  \mu&:\hat{\cT}\otimes \hat{\cT}\rightarrow \hat{\cT}: \hat{T}\otimes\hat{T}'\mapsto \tapGT{\hat{T}}{\hat{T}'}{}\,,\qquad 
    \tapGT{\hat{T}}{\hat{T}'}{}
    :=\sum_{\mu\in \tMatchGT{T}{T'}{}}\delta\left(\left[
      \TcompGT{T}{\mu}{T'}{}
    \right]_{\equiv_N} 
    \right)
  \end{split}\\
  \eta&:\bK\rightarrow \hat{\cT}: k\mapsto k\cdot \hat{T}_{\mIO}\,.
  \end{align}
  Both definitions are suitably extended by (bi-)linearity to generic (pairs of) elements of $\hat{\cT}$.
\end{definition}

\begin{proposition}\label{prop:TraceletAlgebra}
  The morphisms $\mu$ and $\eta$ define an \emph{associative, unital $\bK$-algebra} $(\hat{\cT},\mu,\eta)$, which we refer to as \emph{tracelet algebra}.
\end{proposition}

\begin{definition}[Tracelet coproduct and counit]
  Fixing the \emph{notational convention} $\uplus_{i\in\emptyset}T_i:=T_{\mIO}$ for later convenience, let $T\equiv_N \uplus_{i\in I}T_i$ be the tracelet normal form for a given tracelet $T\in\cT$ (where $T_i\in \mathfrak{Prim}(\cT_N)$ for all $i\in I$ if $T\neq T_{\mIO}$). Then the \emph{tracelet coproduct} $\Delta$ and \emph{tracelet counit} $\varepsilon$ are defined via their action on basis vectors $\hat{T}=\delta(T)$ of $\hat{\cT}$ as
  \begin{align}
    \Delta&: \hat{\cT}\rightarrow \hat{\cT}\otimes \hat{\cT}:\hat{T}\mapsto
    \Delta(\hat{T}):=\sum_{X\subset I}\delta\left(\left[\uplus_{x\in X}T_x\right]_{\equiv_N}\right)\otimes\delta\left(\left[\uplus_{y\in I\setminus X}T_y\right]_{\equiv_N}\right)
  \end{align}%
and $\varepsilon:\hat{\cT}\rightarrow \bK: \hat{T}\mapsto \mathit{coeff}_{\hat{T}_{\mIO}}(\hat{T})$. Both definitions are extended by linearity to generic elements of $\hat{\cT}$.
\end{definition}

\begin{proposition}
    The data $(\hat{\cT},\Delta,\varepsilon)$ defines a \emph{coassociative, cocommutative and counital coalgebra}.
\end{proposition}
\begin{proof}
  Since the construction of $\Delta$ and $\varepsilon$ is the standard construction for a deconcatenation coalgebra (cf.\ e.g.\ \cite{manchon2008hopf}), the proof is omitted here for brevity.
\end{proof}

The algebra and coalgebra structures on $\hat{\cT}$ are compatible in the following sense:

\begin{theorem}[Bialgebra structure]\label{thm:Tbialg}
  The data $(\hat{\cT},\mu,\eta,\Delta,\varepsilon)$ defines a
  \emph{bialgebra}. 
\end{theorem}

By virtue of the definition of the tracelet normal form, it is evident that
both composition and decomposition of tracelets is compatible with a
filtration structure given by the number of ``connected components''
in the following sense:
\begin{theorem}[Compare~\cite{bdgh2016}, Sec.~3.4 and Thm.~3.2]\label{thm:Tfiltration}
The tracelet bialgebra $(\hat{\cT},\mu,\eta,\Delta,\varepsilon)$ is 
filtered by
\begin{equation}
\hat{\cT}^{(n)}:=\mathit{span}_{\bK}\left.\left\{
\hat{T}_1\uplus\dotsc\uplus \hat{T}_n
\right\vert \hat{T}_1,\dotsc,\hat{T}_n\in \mathit{Prim}(\hat{\cT})
\right\}\,,
\end{equation}
and for this filtration it is connected 
($\hat{\cT}^{(0)}:=\mathit{span}_{\bK}\{\hat{T}_{\mIO}\}$).
In particular, it acquires an antipode and becomes a Hopf algebra.
\end{theorem}

Finally, yet again taking inspiration from~\cite{bdgh2016}, one may demonstrate that the tracelet Hopf algebra is isomorphic to a Hopf algebra that is well-known in the setting of the Heisenberg--Weyl diagram Hopf algebra and the Poincar\'{e}--Birkhoff--Witt theorem for ``normal-ordering'' of elements of the Hopf algebra:
\begin{theorem}
Let $\cL_{\cT}:=(\mathit{Prim}(\hat{\cT}),[.,.]_{\diamond})$ denote the \emph{tracelet Lie algebra}, where $[\hat{T},\hat{T}']_{\diamond}:=\hat{T}\diamond\hat{T}'-\hat{T}'\diamond\hat{T}$ is the \emph{commutator operation} (w.r.t.\ $\diamond$). Then the tracelet Hopf algebra is isomorphic (in the sense of Hopf algebra isomorphisms) to the \emph{universal enveloping algebra} of $\cL_{\cT}$.
\end{theorem}

\begin{figure}
\centering
\subfigure[1-simplices\label{fig:r1s}]{
\href{http://nicolasbehr.com/files/asymptote3d/RDS1simplex.html}{%
\includegraphics[scale =0.7]{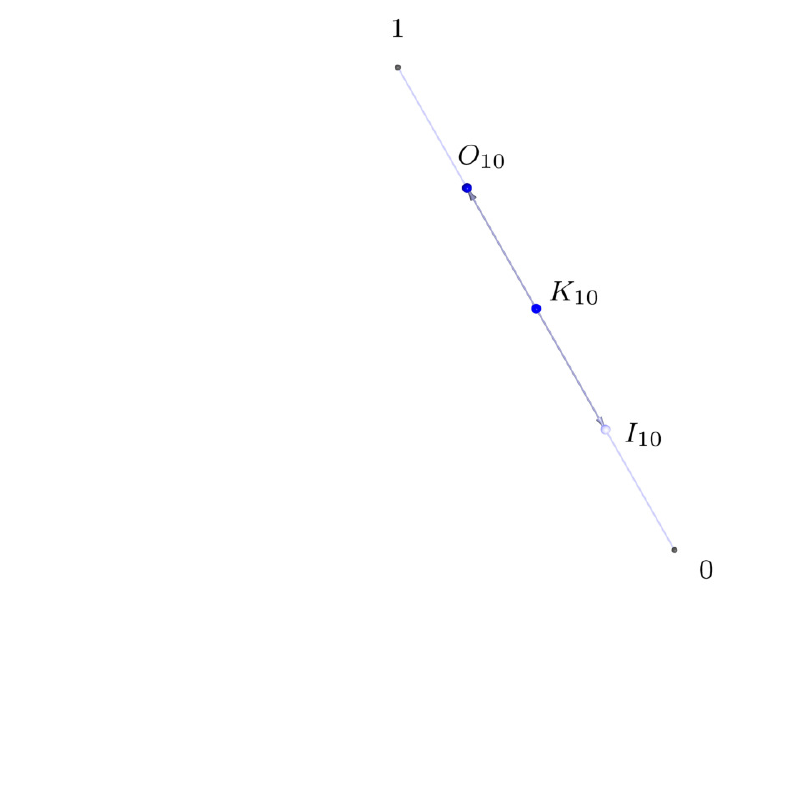}}
}$\qquad\qquad$ \subfigure[2-simplices\label{fig:r2s}]{
\href{http://nicolasbehr.com/files/asymptote3d/RDS2simplex.html}{%
\includegraphics[scale =0.7]{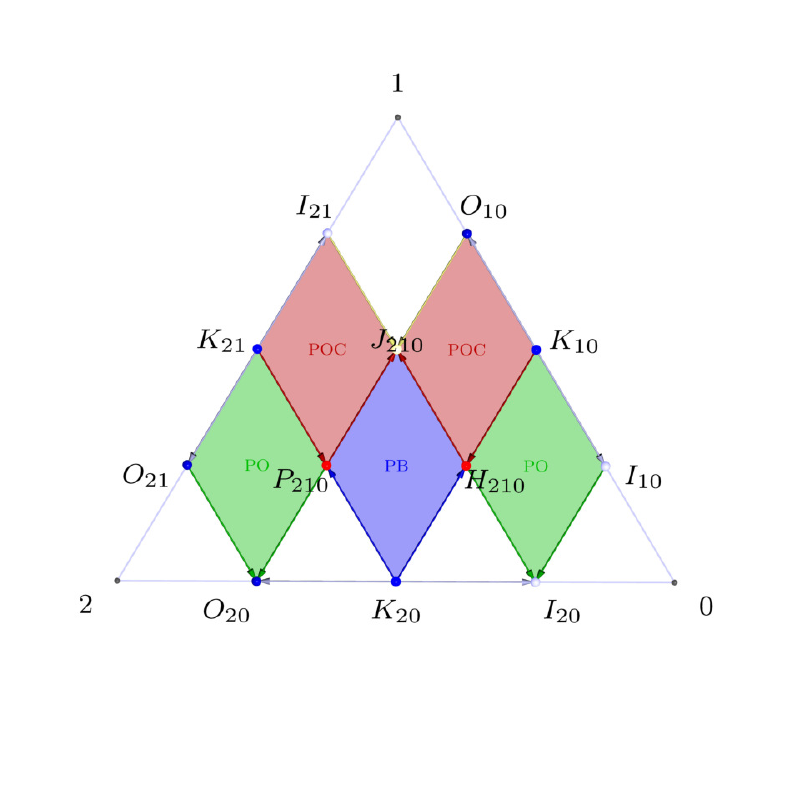}}
}\\
\subfigure[From adjacent 2-simplices\ldots \label{fig:rds3-prepA}]{
\href{http://nicolasbehr.com/files/asymptote3d/RDS2to3simplexA.html}{%
\includegraphics[scale =0.7]{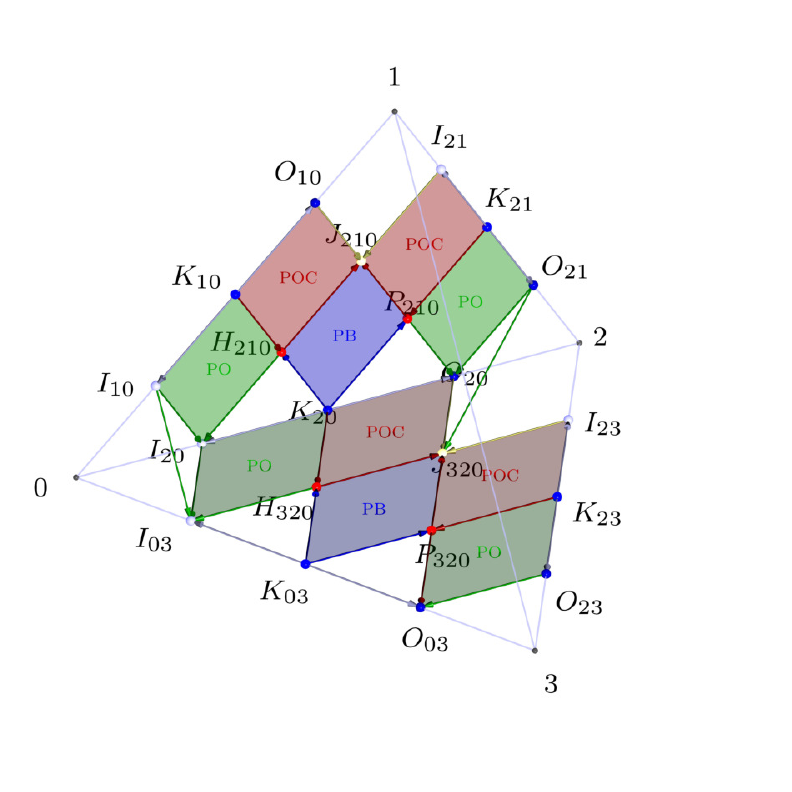}}
}$\qquad\qquad$ \subfigure[\ldots via the Concurrency Theorem\ldots\label{fig:rds3-prepB}]{
\href{http://nicolasbehr.com/files/asymptote3d/RDS2to3simplexB.html}{%
\includegraphics[scale =0.7]{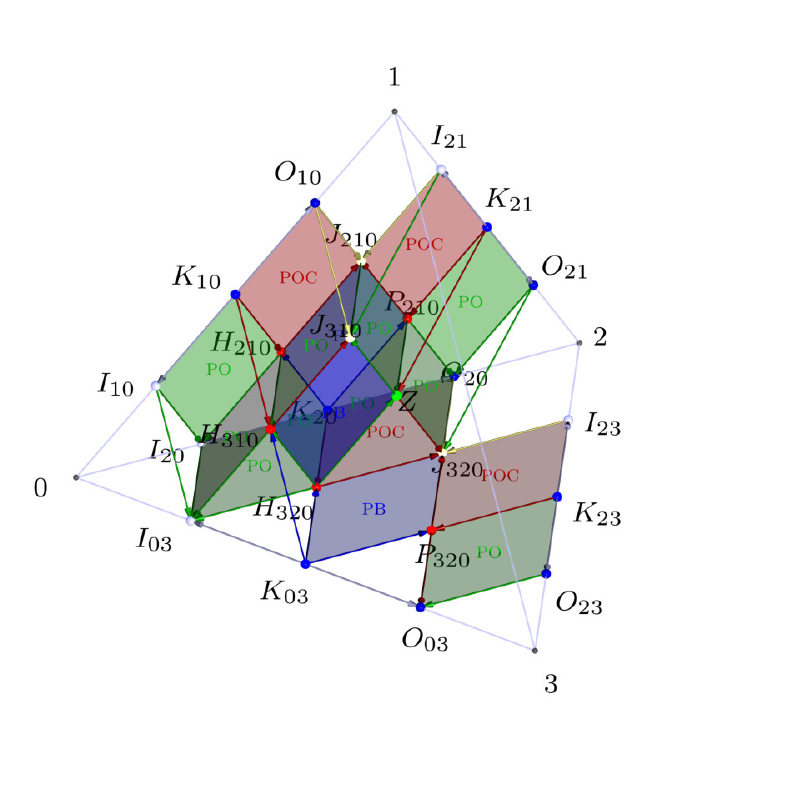}}
}\\
\subfigure[\ldots to 3-simplices.\label{fig:r3s}]{
\href{http://nicolasbehr.com/files/asymptote3d/RDS3simplex.html}{%
\includegraphics[scale =0.7]{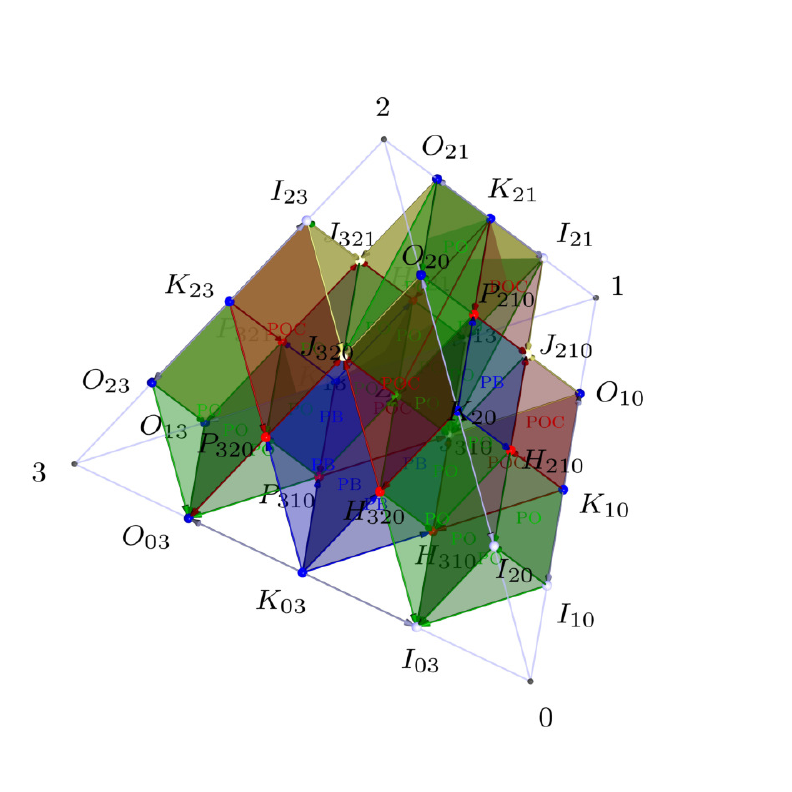}}
}$\qquad\qquad$ \subfigure[Tracelets of length 3 and evaluation\label{fig:rTot}]{
\href{http://nicolasbehr.com/files/asymptote3d/Tracelet-3simplex-and-evaluation.html}{%
\includegraphics[scale =0.7]{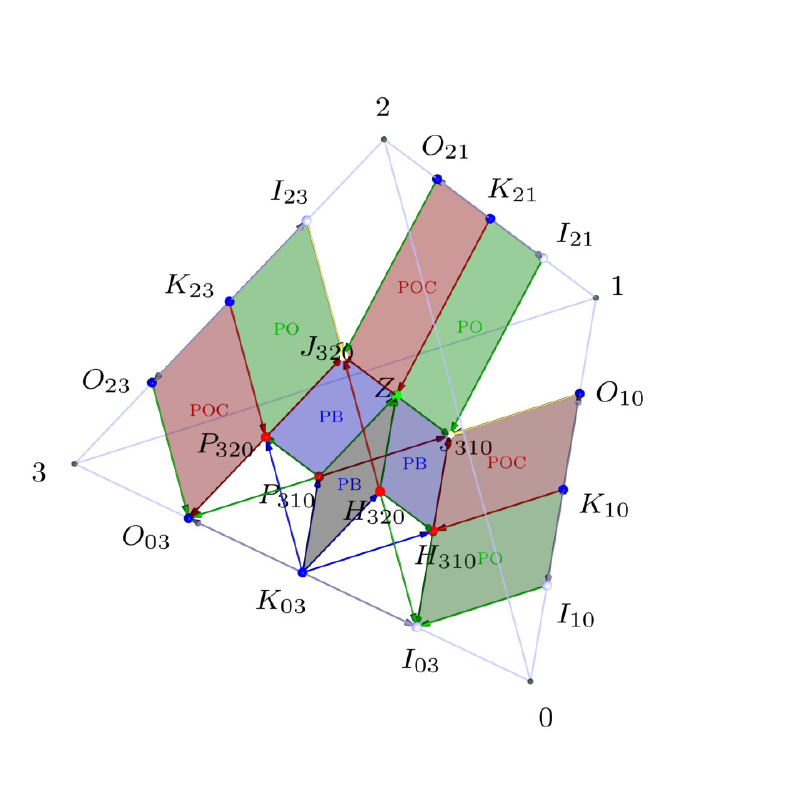}}
}
\caption{Elements of Tracelet Decomposition Space theory (\emph{Note:} in order to allow for a more in-detail inspection, the figures are hyperlinked to on-line interactive 3D views of the respective diagrams).}
\end{figure}

\end{document}